\documentclass[ letterpaper, 11pt]{article}
\usepackage{amsmath,amssymb,amsfonts,mathrsfs,amsthm}
\usepackage[all]{xy}
\usepackage{url}
\usepackage{float}
\usepackage{framed}
\usepackage{enumerate}
\usepackage{color}
\usepackage[normalem]{ulem}
\usepackage{algorithmic}
\usepackage{tikz}
\usetikzlibrary{petri}
\usepackage{xspace}
\usepackage[hang,small,bf]{caption}
\usepackage{fmtcount}
\definecolor{DarkBlue}{rgb}{0.1,0.1,0.5}
\usepackage[colorlinks=true,linkcolor=black,citecolor=DarkBlue]{hyperref}

\usepackage{geometry}
\newtheorem{theorem}{Theorem}[section]
\newtheorem{definition}{Definition}
\newtheorem{lemma}[theorem]{Lemma}
\newtheorem{proposition}[theorem]{Proposition}
\newtheorem{corollary}[theorem]{Corollary}

\newcommand{\eps}{\varepsilon}

\newcommand{\gametree}{\textsf{GameTree}\xspace}
\newcommand{\statetree}{{\textsf{StateTree}}\xspace}
\newcommand{\randomtree}{{\textsf{RandomnessTree}}\xspace}
\newcommand{\potent}{potent\xspace}
\newcommand{\Sim}{{\mathsf{Sim}_\pi}}
\newcommand{\RTBC}{\textsf{RTC}\xspace}
\newcommand{\KTC}{\textsf{SBTC}\xspace}
\newcommand{\tree}{{\cal R}}

\begin{document}

\title{\textbf{Potent Tree Codes and their applications}: \\
Coding for Interactive Communication, revisited}
\author{Ran Gelles  \;\;\; Amit Sahai
 \\ \\ Department of Computer Science, UCLA, Los angeles
\\ \texttt{\{gelles, sahai\}@cs.ucla.edu}}
\date{}
\maketitle

\newcommand{\sectionline}{%
  \nointerlineskip \vspace{\baselineskip}%
  \hspace{\fill}\rule{0.5\linewidth}{.7pt}\hspace{\fill}%
  \par\nointerlineskip \vspace{\baselineskip}
}

\begin{abstract}

In this work, we study the fundamental problem of reliable
\emph{interactive} communication over a noisy channel.  In a
breakthrough sequence of papers published in 1992 and
1993~\cite{schulman92,schulman93}, Schulman gave
\emph{non-constructive} proofs of the existence of general methods
to emulate any two-party interactive protocol such that: (1) the
emulation protocol only takes a constant-factor longer than the
original protocol, and (2) if the emulation protocol is executed
over a noisy channel (BSC), then the probability that the emulation
protocol fails to perfectly emulate the original protocol 
is exponentially small in the total length of the
protocol. Unfortunately, Schulman's emulation procedures either
only work in a model with a large amount of shared
randomness~\cite{schulman92}, or are non-constructive in that they
rely on the existence of \emph{good tree codes}~\cite{schulman93}.
The only known proofs of the existence of good tree codes are
non-constructive, and finding an explicit construction remains an
important open problem. Indeed, randomly generated tree codes are
\emph{not} good tree codes with overwhelming probability.

In this work, we revisit the problem of reliable interactive
communication, and obtain the following results:

\begin{itemize}
\item
We introduce a new notion of goodness for a tree code, and define the notion of a
\emph{\potent tree code}.  We believe that this notion is of
independent interest.
\item
We prove the correctness of an explicit emulation procedure based on
any \potent tree code.  (This replaces the need for good tree codes
in the work of Schulman~\cite{schulman93}.)
\item
We show that a randomly generated tree code (with suitable constant
alphabet size) is a \potent tree code with overwhelming probability.  Furthermore we are able to partially derandomize this result
using only $O(n)$ random bits, where $n$ is the depth of the tree.
\end{itemize}

These (derandomized) results allow us to obtain the first fully
\emph{explicit} emulation procedure for reliable interactive
communication over noisy channels with a constant communication
overhead,  with failure probability that is exponentially small in
the length of the original communication protocol.

Our results also extend to the case of interactive multi-party communication among a constant number of parties.
\end{abstract}

\thispagestyle{empty}   
\newpage
\setcounter{page}{1}


\section{Introduction}

In this work, we study the fundamental problem of reliable
\emph{interactive} communication over a noisy channel.  The famous
coding theorem of Shannon~\cite{shannon48} from 1948 shows how to
transmit any message over a noisy channel with optimal rate such
that the probability of error is exponentially small in the length
of the message.  However, if we consider an interactive protocol
where individual messages may be very short (say, just a single
bit), even if the entire protocol itself is very long, Shannon's
theorem does not suffice.

In a breakthrough sequence of papers published in 1992 and
1993~\cite{schulman92,schulman93}, Schulman attacked this problem
and gave a \emph{non-constructive} proof of the existence of a
general method to emulate any two-party interactive protocol such
that: (1) the emulation protocol only takes a constant-factor longer
than the original protocol, and (2) if the emulation protocol is
executed over a noisy channel (specifically a Binary Symmetric
Channel\footnote{The Binary Symmetric Channel with crossover
probability $p$ is one that faithfully transmits a bit with
probability $1-p$, and toggles the bit with probability $p$.  Note
that Schulman's results as quoted here extend to the case of any
discrete memoryless channel with constant capacity, as do all of our
results.} with some constant crossover probability less than
$\frac12$), then the probability that the emulation protocol fails
to perfectly emulate the original protocol 
is exponentially small in the total length of the protocol.
Unfortunately, Schulman's 1992 emulation procedure~\cite{schulman92}
either required a nonstandard model in which parties already share a
large amount of randomness before they communicate, where the amount
of shared randomness is quadratic in the length of the protocol to
be emulated, or required inefficient encoding and decoding. On the
other hand, Schulman's 1993 emulation procedure~\cite{schulman93}
 is non-constructive in that it
relies on the existence of \emph{good tree codes}\footnote{ We note,
with apology, that what we are calling a ``good tree code'' is what
Schulman calls a ``tree code.''  We make this change of terminology
because we will introduce an alternative relaxed notion of goodness
for a tree code that will lead to our notion of a ``\potent tree
code.'' }.  The only known proofs of the existence of good tree
codes are non-constructive, and finding an explicit construction
remains an important open problem.  Indeed randomly generated tree
codes are \emph{not} good tree codes with overwhelming probability.

In this work, we revisit the problem of reliable interactive
communication, and give the first fully \emph{explicit} emulation
procedure for reliable interactive communication over noisy channels
with a constant communication overhead,  with failure probability
that is exponentially small in the length of the original
communication protocol\footnote{Here we assume that we know the
length of the protocol in advance.}.  To obtain this result, we do
the following:

\begin{itemize}\addtolength{\itemsep}{-0.5em}
\item
We introduce a new notion of goodness for a tree code, and define
the notion of a \emph{\potent tree code}.  We believe that this
notion is of independent interest.
\item
We prove the correctness of an explicit emulation procedure based on
any \potent tree code.  (This replaces the need for good tree codes
in the work of Schulman~\cite{schulman93}.)  This procedure is
efficient given a black box for efficiently decoding the \potent
tree code.
\item
We show that a randomly generated tree code (with suitable constant
alphabet size) is a \potent tree code with overwhelming probability.
Furthermore, we show that a randomly generated tree code (when
combined with a good ordinary error-correcting code) can be
efficiently decoded with respect to a BSC with overwhelming
probability.

\item
Finally, we are able to partially derandomize the above result using
only $O(n)$ random bits, where $n$ is the depth of the tree, while
maintaining the efficiency of decoding.
\end{itemize}

With the above work done, our result is immediate: Since only $O(n)$
random bits are needed, they can be chosen once and for all, encoded
using an ordinary block error-correcting code, and sent to the other
party.  Then a deterministic procedure can be used to finish the
protocol.

Our result extends to the case of any constant number of parties.
For the case of a super-constant number of parties, however, our
explicit emulation procedure will have a $O(m)$ slowdown for $m$
parties (regardless of the length of the protocol).  A
(non-explicit) emulation procedure based on good tree codes was
given by Rajagopalan and Schulman~\cite{RS94} that achieved a
$O(\log m)$ slowdown in the general case.

Also, another result we obtain relates to the recent work of
Braverman and Rao~\cite{BR10}. They consider whether good tree codes
can be used to improve the result of Schulman for \emph{adversarial}
errors --- which only works if the fraction of errors is below
$1/240$. They obtain a very significant improvement: as long as the
fraction of errors is at most $1/4-\epsilon$, any protocol can be
simulated with only a constant slowdown, using a good tree code over
a constant-size alphabet (the simulation tolerates a $1/8-\epsilon$
error fraction when using a binary alphabet). We show that a similar
result can be obtained replacing the good tree code with a \potent
tree, showing our notion is useful even for the case of arbitrary
(adversarial) errors. However, in this case, like all previous work
on the adversarial error case, we do not know how to obtain
efficient decoding against adversarial errors.

\paragraph{Our approach.}
We begin our investigation by asking the question: What properties
does a tree code need in order to be useful for emulating protocols
over noisy channels?  (Without loss of generality, assume that
protocols only exchange one bit at a time from each party.) For the
purpose of this paper, a tree code is simply any deterministic
on-line encoding procedure in which each symbol from the input
alphabet $\Sigma$ is (immediately) encoded with a single symbol from
the output alphabet $S$, but the encoding of future input symbols
can depend on all the input symbols seen so far.  As such, any such
deterministic encoding can be seen as a complete $|\Sigma|$-ary tree
with each edge labeled with a single symbol of the output alphabet
$S$.

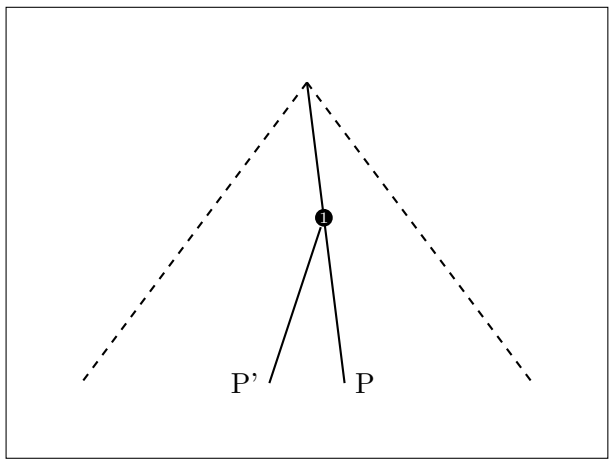
\begin{figure}[htb]
\begin{center}
\begin{tikzpicture}
\coordinate (top) at (4,5);

\draw[thick,dashed] (top) -- (1,1);
\draw[thick,dashed] (top) -- (7,1);
\draw[thick] (top) -- (4.5,1) node [right] {P}
  node (1) [pos=0.3] {}
  node (2) [pos=0.45,token] {1}
  node (3) [pos=0.6] {}
  node (4) [pos=0.75] {}
  node (5) [pos=0.9] {};

  \draw[-,thick] (2) -- (3.5,1) node [left] {P'} ;

\draw (0,0) rectangle (8,6);

\end{tikzpicture}
\end{center}
\caption{A very bad tree code}
\label{fig:verybad}
\end{figure}
The usefulness of some kind of tree code for protocol emulation
seems immediate, since each party must encode the bit it needs to
send, before knowing what other bits it needs to send later (which
it will not know until it receives messages from the other party).
Let us associate every path from the root to a node in the tree code
with the concatenation of output symbols along that path.  Then, at
first glance, it may appear that all we need from the tree code is
for ``long-enough'' divergent paths to have large relative Hamming
distance.  That is, suppose that the tree code illustrated in
Figure~\ref{fig:verybad} has the property that the relative Hamming
distance between the path from node 1 to P and the path from node 1
to P' is very small, even though each of those paths is long.  This
would certainly be problematic since the protocol execution
corresponding to each path could be confused for the other.  As long
as all long divergent paths had high Hamming distance, however, it
seems plausible that eventually the protocol emulation should be
able to avoid the wrong paths.  Also, it is important to note that
with suitable parameters, a randomly generated tree code would
guarantee that all long divergent paths have high relative Hamming distance with overwhelming probability.

However, this intuition does not seem to suffice, because while the protocol emulation is proceeding down an \emph{incorrect} path, one party is sending the \emph{wrong} messages -- based on wrong interpretations of the other party's communication.  After a party realizes that it has made a mistake, it must then be able to ``backtrack'' and correct the record going forward.  The problem is that even short divergent paths with small relative Hamming distance can cause problems.  Consider the tree code illustrated in Figure~\ref{fig:bad}.  In this figure suppose the path along the nodes 1, 2, and 3 is the ``correct'' path, but that the short divergent paths from 1 to A, 2 to B, and 3 to C all have small relative Hamming distance to the corresponding portions of the correct path.  Then in the protocol emulation, because of the bad Hamming distance properties, the emulation may initially incorrectly proceed to node A, and then realize it made a mistake.  But instead of correctin!
 g to a node on the correct path, it might correct to the node A' and proceed down the path to B.  Then it may correct to B', and so on.  Because the protocol emulation keeps making mistakes, it may never be able to successfully backtrack and communicate the messages that correspond to the actual protocol execution.
\begin{figure}[htb]
\begin{center}
\begin{tikzpicture}
\coordinate (top) at (4,5);

\draw[thick,dashed] (top) -- (1,1);
\draw[thick,dashed] (top) -- (7,1);
\draw[thick] (top) -- (4.5,1)
  node (1) [pos=0.3,token] {1}
  node (2) [pos=0.45,token] {2}
  node (3) [pos=0.6,token] {3}
  node (4) [pos=0.75,token] {}
  node (5) [pos=0.9,token] {};

  \draw[-,thick] (1) -- ++(-50:1.2cm) node (a) [token,fill=gray,label=right:A]{};
  \draw[-,thick] (2) -- ++(-120:1.2cm) node (b) [token,fill=gray,label=left:B] {};
   \draw[-,thick] (3) -- ++(-50:1.2cm) node[token,fill=gray,label=right:C]{};

   \draw[-,dashed] (a) -- +(-0.9,0) node [token,fill=red,label=left:A'] {};
   \draw[-,dashed] (b) -- +(1.1,0)  node [token,fill=red,label=right:B'] {};

   \draw[-, thick] (4) -- ++(-120:1.1cm);
   \draw[-, thick] (5) -- ++(-50:0.5cm);

\draw (0,0) rectangle (8,6);

\end{tikzpicture}

\end{center}
\caption{A bad tree code}
\label{fig:bad}
\end{figure}
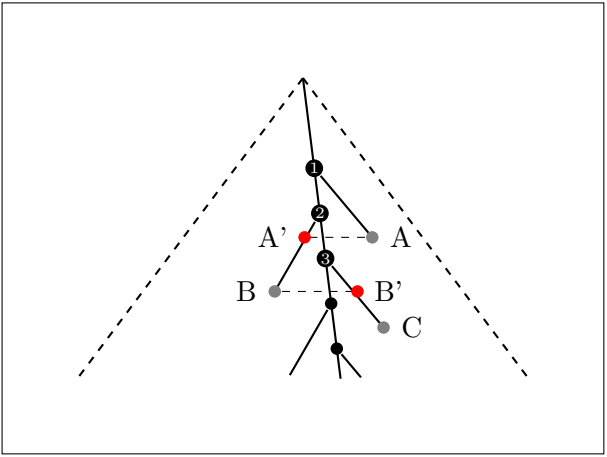

Schulman~\cite{schulman93} dealt with this problem by simply insisting that \emph{all} divergent paths have large relative Hamming distance in his definition of a good tree code.  This would prevent all such problems, and guarantee that errors in emulation could only be caused by actual channel errors.  The downside of this approach is that randomly generated tree codes would have short divergent paths with small (even zero) relative Hamming distance with overwhelming probability, and thus would not be good tree codes.

Our main observation is that this requirement goes too far.  If a tree code has the property that for every path from root to leaf, there are only a few small divergent branches with low relative Hamming distance (as illustrated in Figure~\ref{fig:potent}), then the emulation protocol will be able to recover from these few errors without any problems.  We call such tree codes \emph{\potent tree codes} since they are sufficiently powerful to enable efficient and reliable interactive communication over a noisy channel.

More precisely, let $\epsilon$ and $\alpha$ be two parameters from the interval $[0,1]$.
Define a path from node $u$ to a descendant node $v$ (of length
$\ell$) to be \emph{$\alpha$-bad} if there exists a path from $u$ to
another descendant node $w$ (also of length $\ell$) such that the
Hamming distance between the $u$-$v$ path and the $u$-$w$ path is less
than $\alpha \ell$.
Then an \emph{$(\epsilon,\alpha)$-potent tree code} of depth $n$ is such that for every path $Q$ from root to leaf, the number of nodes in the union of all $\alpha$-bad subpaths of $Q$ is at most
$\epsilon n$.

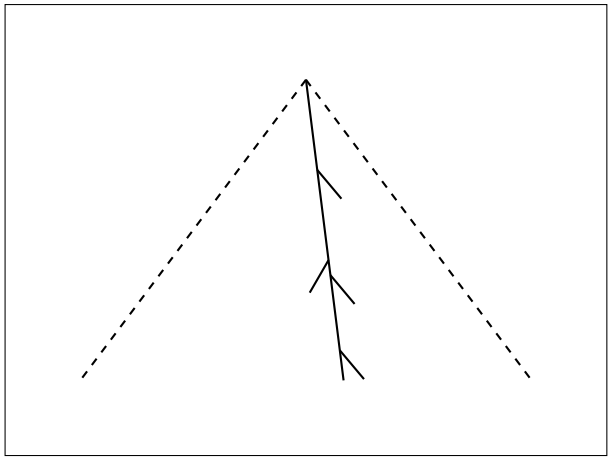
\begin{figure}[htb]
\begin{center}
\begin{tikzpicture}
\coordinate (top) at (4,5);

\draw[thick,dashed] (top) -- (1,1);
\draw[thick,dashed] (top) -- (7,1);
\draw[thick] (top) -- (4.5,1)
  node (1) [pos=0.3] {}
  node (2) [pos=0.45] {}
  node (3) [pos=0.6] {}
  node (4) [pos=0.65] {}
  node (5) [pos=0.9] {};

  \draw[-,thick] (1.center) -- ++(-50:0.5cm) ;
  \draw[-,thick] (3.center) -- ++(-120:0.5cm);
     \draw[thick] (4.center) -- ++(-50:0.5cm);

   \draw[thick] (5.center) -- ++(-50:0.5cm);

\draw (0,0) rectangle (8,6);

\end{tikzpicture}
\end{center}
\caption{A potent tree code}
\label{fig:potent}
\end{figure}

We show that randomly generated tree codes (with suitable constant
alphabet sizes) are potent tree codes with overwhelming probability.
As hinted above, because every root-leaf path has good properties, a
potent tree code will work for emulating \emph{any} (adversarially
chosen) interactive protocol.  With some additional randomization,
we show that within such emulations, decoding of a randomly
generated potent tree code can be done efficiently even for an
adversarially chosen protocol.

\paragraph{Naturalness of our definition.}
We argue that our notion of bad subpaths in a \potent tree code
captures the level of local ``confusion'' that is possible in a tree
code, in a manner that we see as analogous to how ordinary symbol
overlap (Hamming ``closeness'') captures such confusion in the
context of ordinary error-correcting codes, which are much less
structured objects.  In this analogy, \potent tree codes with
$\epsilon=0$ (which correspond to Schulman's good tree codes) are
analogous to maximum distance separable (MDS) codes in the context
of ordinary error-correcting codes.  Just as MDS codes are powerful
and useful objects, but not necessary for most applications of
error-correcting codes, we think of Schulman's good tree codes as
being powerful and useful objects, but not necessary for important
applications like reliable interactive communication where \potent
tree codes suffice.

\paragraph{Other Related Work.}
In 2006, Peczarski~\cite{peczarski06} provides a randomized way for
constructing good tree codes. The construction succeeds with
probability $1-\epsilon$ using alphabet with size proportional to
$\epsilon^{-1}$. Therefore, using Peczarski's method to construct a
good tree code with exponentially small failure probability
$\epsilon$, yields a polynomial slowdown; or a sub-linear but
super-logarithmic slowdown if  $\epsilon$ is negligible (in the
length of the simulated protocol). Other methods for constructing a
good tree code are reported by Schulman~\cite{schulman-email}, yet
they require polynomial-size alphabet (in the depth of the tree),
resulting in a logarithmic slowdown using Schulman's
emulation~\cite{schulman93}. Schulman~\cite{schulman-email} also
provides methods for constructing tree codes with weaker properties
such as satisfying the Hamming distance property for only a
logarithmic depth (which yields a failure probability that is
inverse-polynomial). Ostrovsky, Rabani, and Schulman~\cite{ORS05}
consider a relaxed problem of communication for control of
polynomially bounded systems, and gave explicit constructions of
codes suitable for that setting.

In work concurrent and independent to ours, Moitra~\cite{Moitra11}
introduced a relaxation of good tree codes that he calls local tree
codes, which allows him to obtain a fully explicit and deterministic
emulation protocol, but which obtains error probability that is a
fixed inverse polynomial in the length of the protocol.  In
contrast, our work obtains a fully explicit randomized emulation
protocol, but achieves error probability that is exponentially small
in the length of the protocol.

\section{Preliminaries}\label{sec:pre}

We begin with several definitions that we use later.
Unless otherwise mentioned, we use base 2 for all logarithms.
\begin{definition}
We say that a function $f(n)$ is negligible in $n$, and denote $f<neg(n)$
if for any polynomial $P$, and sufficiently large  $n$,
$f(n)< \frac{1}{P(n)}$.
\end{definition}
Our model of communication is based on a binary channel that flips each
bit with probability $p_{BSC}$, independently of other bits.
\begin{definition}
A binary symmetric channel (BSC) with error probability $p_{BSC}$
is a binary channel $\{0,1\}\to\{0,1\}$ such that for every inputed bit
outputs the same bit with probability $1-p_{BSC}$ or the complementary bit
with probability $p_{BSC}$, independently of previous transmissions (memoryless).
\end{definition}

One can use codes in order
to send messages which can be recovered except with arbitrary small probability.
This is done by adding redundancy to each message, according to the desired
error probability.
Shannon's coding theorem asserts the existence of an error-correcting code that reduces the
error probability (for a single message) to be exponentially small,
while increasing the amount of transmitted information
by only a constant factor.

\begin{lemma}[Shannon Coding Theorem~\cite{shannon48}]\label{lem:shannon}
For a BSC channel with capacity $C$, an alphabet $S$ and any $\xi>0$, there exists a code
$enc: S\to \{0,1\}^n$ and $dec: \{0,1\}^n \to S$ with $n=O(\tfrac1C\xi\log|S|)$ such that
\[
\Pr \left [dec ( BSC ( enc(m) ) ) \ne m \right] < 2^{-\Omega(\xi\log |S|)}.
\]
\end{lemma}
\noindent
Although throughout this paper we assume the channel is a BSC with symbol error of at most $p$ (using an error correction code), our result applies for any memoryless noisy channel with maximal symbol error probability $p$.

The main structure we use is a \emph{tree code}, introduced by Schulman~\cite{schulman93,schulman96}.
Each edge in a tree code is assigned with a label (from a given alphabet $S$),
such that strings obtained by concatenation of these labels
form a code, that is, have a large Hamming distance.

\begin{definition}
The Hamming distance of two strings $\sigma=\sigma_1\ldots\sigma_m$ and
$\sigma'=\sigma'_1\ldots\sigma'_m$ of the same length over an alphabet $S$,
is the number of positions $i$ such that $\sigma_i \ne \sigma'_i$.
The Hamming distance is denoted by $\Delta(\sigma,\sigma')$.
\end{definition}

As said earlier, we re-define the term \emph{tree code} to be any tree, such that
every arc $i$ in the tree has a label $\sigma_i$ over some fixed alphabet $S$.
Denote with $w(s)$ the label of the arc between
$s$ and its parent and $W(s)$ the concatenation of the labels along the
route from the root to the node $s$. Using our new terminology,
the tree codes introduced by Schulman~\cite{schulman96}
are denoted as \emph{good tree codes}.
\begin{definition}[Tree Codes~\cite{schulman96}] \label{def:TreeCode}
A  \textbf{good $d$-ary tree code} over an alphabet $S$, of distance parameter $\alpha$ and depth $n$,
is a $d$-ary tree code of depth $n$ such that
for every two nodes $s$ and $r$ at the same depth,
$$\Delta(W(s),W(r)) \ge \alpha l\text{,}$$
where $l$ is the distance from $s$ and $r$ to their least common ancestor.
\end{definition}
Tree codes can be used to communicate a node $u$ between the users, by sending the labels
$W(u)$. Decoding a transmission means recovering the node at the end of
the route defined by the received string of labels.
In order to reduce the error probability of the label
transmission, each label is separately coded using a standard error-correcting code.
It is shown in~\cite{schulman96} that for every distance parameter $\alpha\in(0,1)$, there exists a good $d$-ary tree code of infinite depth, labeled using  $|S|\le 2\lfloor (2d)^{\frac{1}{1-\alpha}}\rfloor-1$ symbols.
However, although it is known to exist, its explicit efficient construction
remains an open question.

\section{Potent Tree Codes}\label{sec:potent} 

\subsection{Potent Tree Codes and Their Properties}\label{sec:potent}

We now formally define the set of \emph{\potent trees} and its complement,
the set of \emph{bad trees}. The latter contains trees that are not useful for our purpose:
at least one of their paths is composed of
``too many'' sub-paths that do not satisfy the distance condition, i.e.,
the total length of these sub-paths
is at least $\varepsilon$ fraction of the tree depth $N$, for some fixed constant $\varepsilon>0$. Formally,
\begin{definition}\label{def:badNode}
Let $u,v$ be some nodes at the same depth $h$ of a tree-code, and let
$w$ be their least common ancestor, located at depth $h-\ell$.
The nodes $u$ and $v$ are
$\alpha$-\textbf{bad nodes} (of length $\ell$)
if $\Delta(W(u),W(v)) < \alpha\ell$.
In this case, the path (of length $\ell$) between $w$ and $u$ is called an $\alpha$-\textbf{bad path} (similarly, the path between $w$ and $v$ would also be a bad path).
Define the imposed $\alpha$-\textbf{bad interval} (of length $\ell$)
as the interval $[h-\ell, h]$.
\end{definition}

\begin{definition}\label{def:badTree}
An \textbf{$(\varepsilon,\alpha)$-bad tree}
is a tree of depth $N$
that has a path containing $\alpha$-bad subpaths, such that their
union is of total length at least $\varepsilon N$.
\end{definition}
\begin{definition}
An  \textbf{$(\eps,\alpha)$-\potent tree code}  is a tree of depth $N$,
such that for every path $Q$ from root to leaf,
the union of all bad subpaths of $Q$ is of length less than $\eps N$.
In other words,  the tree is \emph{not} an $(\eps,\alpha)$-bad tree.
\end{definition}
\noindent We stress that a bad tree is not necessarily bad in \emph{all} of its paths,
since the existence of a single bad path is sufficient.

Conveniently, it is rather simple to construct a \potent tree,
which makes it a feasible
tool for plenty of applications.
In the following we give two methods for constructing \potent trees.
The straightforward method is to randomly pick each label of the tree.
The obtained Random Tree Code (\RTBC) is a potent tree except with probability
exponentially small in the depth of the tree.
The drawback of the first construction, is that its description is exponential.
However, we observe that our proof does not require the entire tree to be random,
but rather makes a use of the fact that any two paths along the tree are \emph{independent}.
Using the method of Alon, Goldreich, H{\aa}stad and Peralta~\cite{AGHP92}
we are able to construct a  tree in which any two paths are \emph{almost independent}.
Moreover, such a tree has an efficient description.

\subsection{Random Tree Codes as Potent Trees}

\begin{definition}
Let $S$ be a finite alphabet.
A {\bf random tree code} (\RTBC) is a $d$-ary tree,
where each arc $i$ has a label $\sigma_i\in S$,
randomly and independently chosen.
\end{definition}

\noindent
\textbf{Intuition.}
It is important to note that a \RTBC might {\em not} be a good tree code
according to Definition~\ref{def:TreeCode}.
However, with high probability, the \RTBC can be used to replace a tree code, without significantly
damaging the probability of success.
Informally speaking, we can think of a \RTBC as using a good tree code, but increasing
the channel error rate. Indeed, the only difference is that with a probability of
$1/|S|$ two edges in the \RTBC are assigned with the same label
(and as a consequence some paths might not satisfy the distance condition).
An equivalent result is obtained by taking a good tree code and ``forcing'' the channel
to make an error during the transmissions related to those edges.
This leads to an expected increase of $1/|S|$ in the channel's error rate.

\begin{theorem}\label{thm:RTCisPotent}
Suppose $\eps,\alpha \in (0,1)$.
Except with probability $2^{-\Omega(N)}$,
a \RTBC with alphabet $|S|>(2d)^{(1+2/\eps)/(1-\alpha)}$
is $(\varepsilon,\alpha)$-\potent.
\end{theorem}
\noindent The proof is given in Appendix~\ref{app:RTBC}.
\bigskip

\subsection{Small-Biased Random Trees as Potent Trees}

In order to agree on a \RTBC with alphabet $S$, the users need to communicate
(or pre-share) $O(d^N\log |S|)$ random bits. Surprisingly, we can reduce the description size to
$O(N\log|S|)$ and still have a \potent code with overwhelming probability.
This is allowed due to Alon et al.'s construction of  a sample space with an efficient description
that is $\epsilon$-biased~\cite{AGHP92}.

\begin{definition}[$\epsilon$-biased sample space~\cite{NN90, AGHP92}]
A sample space $X$ on $n$ bits is said to be $\epsilon$-biased with respect to linear tests
if for every sample $x_1\dotsm x_n$ and every string $\alpha_1\dotsm\alpha_n \in \{0,1\}^n\smallsetminus\{0\}^n$, the random variable $y=\sum_{i=1}^n \alpha_ix_i \mod 2$ satisfies
$|\Pr[y=0]-\Pr[y=1]| \le \epsilon$.
\end{definition}
We use~\cite[Construction 2]{AGHP92} to achieve a sample space $\mathbf B_n$
which is $\epsilon$-biased with respect to linear tests.
Let $p$ be an odd prime such that $p>(n/\epsilon)^2$,
and let $\chi_p(x)$ be the quadratic character of {$x$ (mod~$p$)}.
Let $\mathbf B_n$ be the sample space described by the following construction.
A point in the sample space is described by a number $x\in [0,1,\ldots, p-1]$,
which corresponds to the $n$-bit string $r(x)=r_0(x)r_1(x)\dotsm r_{n-1}(x)$
where $r_i(x) = \frac{1-\chi_p(x+i)}{2}$.
\begin{proposition}[\cite{AGHP92}, Proposition~2]
The sample space $\mathbf B_n$ is $\frac{n-1}{\sqrt p}+\frac{n}{p}$-biased with respect to linear tests.
\end{proposition}
We use the above to construct
a $d$-ary tree code of depth $N$ with labels over an alphabet $S$.
Without loss of generality we assume that $|S|$ is a power of 2, and
describe the tree as the $d^N\log|S|$-bit string constructed by concatenating of all
the tree's labels in some fixed ordering.
Since each $n$-bit sample describes a tree-code,
we are sometimes negligent with the distinction between these two objects.

\begin{definition}\label{def:SBTC}
A $d$-ary \emph{Small-Biased Tree Code} (\KTC) of depth $N$, is a tree described by some sample
from the sample space $\mathbf B_n$ with $n=d^N\log |S|$, $\epsilon = 1/2^{cN\log |S|}$ for some constant $c$ that we can choose later.
\end{definition}

We note that small-bias trees have several properties which are very useful for our needs.
Specifically, every set of labels are almost independent.
\begin{definition}[almost $k$-wise independence~\cite{AGHP92}]\label{def:independent}
A sample space on $n$ bits is \textbf{$(\epsilon,k)$-independent} if for any
$k$ positions $i_1 < i_2 < \dotsm < i_k$ and $k$-bit string $\xi$,
\[
\lvert \Pr[ x_{i_1}x_{i_2}\dotsm x_{i_k} = \xi] - 2^{-k}\rvert \le \epsilon
\]
\end{definition}
Due to a lemma by Vazirani~\cite{Vazirani86} (see also corollary~1 in~\cite{AGHP92}),
if a sample space is $\epsilon$-biased with respect to linear tests, then for every $k$, the sample space is $((1-2^{-k})\epsilon, k)$-independent. Thus, $\mathbf B_n$ is $(\epsilon,k)$-independent, for any $k$.
\begin{corollary}\label{cor:k-independence}
Let ${\cal T}$ be a $d$-ary \KTC of depth $N$, then any $1\le k\le d^N$ labels of ${\cal T}$ are almost independent,
that is, any $k\log|S|$ bits of ${\cal T}$'s description
are $(2^{-cN\log|S|},k)$-independent.
\end{corollary}

Finally, let us argue that such a construction is efficient.
Let $p=O((n/\epsilon)^2)$ and assume a constant alphabet $|S|=O(1)$.
Each sample $x$ takes $\log p=O(N)$ bits, and each $r_i(x)$ can be computed by $poly(N)$ operations.

\bigskip
We now show that the properties shown in Appendix~\ref{app:RTBC} for a \RTBC,
hold for a \KTC as well.
\begin{proposition}\label{prop:KTCisPotent}
Suppose $\eps,\alpha \in (0,1)$.
Except with probability $2^{-\Omega(N)}$,
a \KTC of depth $N$ over alphabet $|S|>(2d)^{(2+2/\eps)/(1-\alpha)}$
is $(\varepsilon,\alpha)$-\potent.
\end{proposition}
\begin{proof}
We show that the probability of a \KTC to be
$(\varepsilon,\alpha)$-bad is exponentially small.
We begin by fixing a leaf $z$, and later use a union bound
to bound the probability over the entire tree.
Assume that the tree is bad, that is,
there exist bad intervals of total length $\varepsilon N$.
Due to Lemma~\ref{lem:intervals}
there must exist \emph{disjoint}
bad intervals of total length at least $\varepsilon N/2$.

There are at most $\sum_{j=\varepsilon N/2}^N{N \choose j} \le 2^{N}$
ways to distribute these disjoint
intervals along the path from  root to $z$.
In a similar way to Lemma~\ref{lem:prob4smallHamming},
we can bound the probability of  having
a node $u$ at the same depth as $z$ which imposes a bad interval of length $\ell$.
Since the tree is $(1/2^{cN\log|S|}, 2l\log|S|)$-independent,
the suffixes (of length $l$) of the label sequences $W(u)$ and $W(z)$ are almost independent.
For $l>0$ and a node $u$  denote by $W_l(u)$ the last $l$ symbols of $W(u)$.
\begin{lemma}\label{lem:ktcHD}
For any two nodes at the same level $z,u$ with a common ancestor $l$ levels
away,
$\Pr [\Delta (W(u),W(z)) = j ] \le {l \choose l-j} \left(\frac{1}{|S|}\right)^{l-j} + 2^{-\Omega(N)}$
\end{lemma}
\begin{proof}
Note that $W(u)$ and $W(v)$ are identical except for their suffix of length $l$.
\begin{align*}
\Pr [\Delta(W(u),&W(v))=j]  = \\
&\sum_{\xi_u,\xi_v} \Pr[W_l(u)=\xi_u, W_l(v)=\xi_v] \Pr [\Delta(W_l(u),W_l(v))=j \mid W_l(u)=\xi_u, W_l(v)=\xi_v]  \\
&\le (2^{-2l\log|S|}+2^{-cN\log|S|})\sum_{\xi_u,\xi_v} \Pr [\Delta(W_l(u),W_l(v))=j \mid W_l(u)=\xi_u, W_l(v)=\xi_v] \\
&\le (2^{-2l\log|S|}+2^{-cN\log|S|}) 2^{2l\log|S|} {l \choose l-j} \left(\frac{1}{|S|}\right)^{l-j}\left(\frac{|S|-1}{|S|}\right)^{j}
\end{align*}
Choosing $c>3$ completes the proof.
For the ease of notation, in the following we use $2{l \choose l-j} \left(\frac{1}{|S|}\right)^{l-j}$ as an upper bound of the above probability.
\end{proof}
The above lemma leads to the following bound on the probability that two nodes are $\alpha$-bad.
\begin{corollary}\label{cor:defect}
\begin{align*}
\Pr [\Delta (W(u),W(z)) \le \alpha l ] &=
\sum_{j=0}^{\alpha l} \Pr [\Delta (W(u),W(z)) = j] \\
&\le \sum_{j=0}^{\alpha l} 2{l \choose l-j} \left(\frac{1}{|S|}\right)^{l-j}
\le 2\frac{2^l}{|S|^{(1-\alpha)l}}\text{ .}
\end{align*}
\end{corollary}
\noindent Using a union bound, the probability that there exist a node
$u\ne z$ with common ancestor $l$ level away, such that
$z$ and $u$ do not satisfy the distance requirement is bounded by
$\sum_u 2\frac{2^l}{|S|^{(1-\alpha)l}}=2(2d/|S|^{1-\alpha})^l$

Consider again the path from root to $z$, and the disjoint
bad intervals of total length at least $\varepsilon N/2$ along it.
There are at most $2N$ labels involved (along both the path to $z$ and the colliding paths).
Since the intervals are disjoint, their probabilities are almost independent as well, and
the probability that a specific pattern of interval happens is bounded by the
multiplication of the probabilities of each interval.

According to the above, 
the probability for a \KTC to be $(\varepsilon,\alpha)$-bad is bounded by
\begin{align*}
\Pr[\text{ \KTC is $(\varepsilon,\alpha)$-bad }] &\le \sum_{z}
\sum_{\stackrel{\ell_1,\ell_2, \ldots \text{ disjoint,}}{\text{ of length} \ge \eps N/2}}
\prod_i 2(2d/|S|^{1-\alpha})^{\ell_i}  \\ &
\le d^N \cdot 2^N (4d/|S|^{1-\alpha})^{\sum_i \ell_i} \le (2d)^N (4d/|S|^{1-\alpha})^{\eps N/2}
\end{align*}
which is exponentially small in $N$ for $|S|> (4d\cdot(2d)^{2/\eps})^{1/(1-\alpha)}$.
\end{proof}

\section{Applications - Simulation with Adversarial Errors}\label{sec:appBR}
In a recent paper~\cite{BR10} Braverman and Rao show how to simulate
any 2-party protocol over a noisy channel, that is able to withstand an error
rate of up to $1/4-\epsilon_2$, for any constant $\epsilon_2>0$.
Their simulation uses good tree codes to communicate
the process of the simulated protocol over the noisy channel.

We show that the analysis of Braverman and Rao
can be repeated using a $(\epsilon_1,1-\epsilon_2)$-potent tree instead of a good tree code,
and withstand error rate of up to $1/4-2\epsilon_1-\epsilon_2$.
Intuitively, for every node which is not $\alpha$-bad, the potent tree code behaves
exactly like a good tree code (i.e., many channel errors are required for having a decoding error).
On the other hand, for every possible path along the potent tree,
there are at most $\epsilon_1N$ nodes which are $(1-\epsilon_2)$-bad, that is,
at most additional $\epsilon_1N$ times in which the scheme differs from a good tree code (in each one of the directions of communication).
This gives an algorithm that withstand up to $1/4-(2\epsilon_1+\epsilon_2)$
fraction of (adversarial) errors.

\begin{theorem}\label{thm:BR}
For any 2-party binary protocol $\pi$ and any constant $\epsilon>0$ there exist a
protocol $\Pi$ that simulates $\pi$ over a noisy channel using \potent tree-codes,
imposes a constant slowdown
and succeeds except with negligible probability.
\end{theorem}
\noindent See proof in Appendix~\ref{app:BR}.

\section{Applications - Efficient Simulation with   Random Errors }\label{sec:app}

We provide an efficient randomized algorithm that succeeds to simulate
any interactive protocol over noisy channel with overwhelming probability.
In 1992 Schulman proposed an efficient randomized scheme that solves this
problem~\cite{schulman92} which requires quadratic communication\footnote{The simulation itself impose a constant slowdown, however the users must share a parity-checking matrix of quadratic size.}.
By using potent trees (realized via \KTC{}s), we improve the result of Schulman and obtain
a linear communication (i.e., a constant dilation) which includes the communication required to agree on the same \KTC.
The scheme we obtain is efficient and constructive.
We then extend our proof to any multiparty protocol following the analysis of Rajagopalan and Schulman~\cite{RS94}, again, by replacing the good tree code with a potent tree.

\subsection{Interactive Protocol Over Noisy Channels}
Our setting considers a distributed computation of a fixed function $f$, performed by several
users who (separately) hold the inputs. We begin by considering only two users and later
extend our result to any number of users.
Let $\pi$ be a 2-party distributed protocol
which on inputs $x_A, x_B$, both parties output the value $f(x_A,x_B)$.
In each round, $A$ and $B$ send a single message to each other, based on their input and
messages previously received. The protocol $\pi$ assumes an ideal communication channel which
contains no errors. Under these assumptions, $\pi$ takes $T$ rounds of communication to output the correct answer, where one round means both users simultaneously send each other a message.

In a more realistic model, the channel between $A$ and $B$ may be noisy, so that each message
needs to be encoded in order to identify and correct possible errors.
Shannon's \emph{Coding Theorem}~\cite{shannon48} (see Lemma~\ref{lem:shannon})
shows that an exponentially small decoding error in the length of the message $|m|$
can be achieved, if the message is encoded into a code word of length $c|m|$,
for some constant $c$
determined by the channel \emph{capacity}.
However, if we use a standard Shannon code to encode multiple messages,
then the probability of having at least a single decoding error
is proportional to the number of messages sent, rather than arbitrarily small.
In this paper we explore the worst case scenario of the above tradeoff between the
number of messages and their size.
Namely, we assume that a total amount of $T$ bits of information  is divided into
$T$ messages of a single bit each. Our aim is to send $O(T)$ bits over the channel
and obtain an exponentially small failure probability.

Let us formulate the computation process of the protocol $\pi$.
During each round, each user $i\in \{A,B\}$ sends one bit according to
its input $x_i$ and the messages received so far. Let $\pi_i(x_i,\emptyset)$
denote the  first bit sent by user $i$, and let $\pi(x,\emptyset)\in \{00,01,10,11\}$ be the
two bits transmitted in the first round by A and B respectively, where $x=x_Ax_B$.
Generally, let $m_1,\ldots,m_t$ be the first $t$ (2 bit-)messages exchanged during the protocol, then the information sent in round $t+1$
is defined by $\pi(x, m_1\ldots m_t)$.

\begin{figure}[htb]
\begin{framed}
\begin{tikzpicture}
\node (root) [above] {\gametree} ;
\coordinate
  child {    edge from parent node [near end,right] {00} }
  child {    edge from parent node [near end,right] {01} }
  child  {
    child {    edge from parent node [near end,right] {00} }
    child {    edge from parent node [near end,right] {01} }
    child {    edge from parent node [near end,right] {10} }
    child {    edge from parent node [near end,right] {11} }
    edge from parent node[near end,right] {10} }
  child {    edge from parent node [near end,right] {11} }
  ;
\end{tikzpicture}
\begin{tikzpicture}
[level/.style={sibling distance=9mm}]   
\node (root) [above] {\statetree} ;
\coordinate
  child {node {00x0}}
  child {node {00x1}}
  child {node {01x0}}
  child {node {01x1}}
  child {node {10x0}
      child child child {edge  from parent [draw=white] node {$\ldots$}}  child child}
  child {node {10x1}}
  child {node {11x0}}
  child {node {11x1}}
  child {node {Hx0}}
  child {node {Hx1}}
  child {node {Bx0}}
  child {node {Bx1}}
  ;
\end{tikzpicture}
\end{framed}
\caption{The \gametree and the \statetree}
\label{fig:trees}
\end{figure}
The computation (over a noiseless channel) can be described as a single route $\gamma_{\pi,x}$
along the \gametree,
a 4-ary tree of depth $T$ (see Figure~\ref{fig:trees}). The path $\gamma_{\pi,x}$ begins
at the root of the tree
and the $t^{\text{th}}$ edge is determined by the 2 bits exchanged in the $t^{\text{th}}$ round,
i.e., the first edge in the path is $\pi(x,\emptyset)$, the second
is $\pi(x, \pi(x,\emptyset))$, etc.

\subsection{Simulating $\pi$ Over a Noisy Channel}

\subsubsection{The basic scheme}\label{sec:basicScheme}
Our goal is to calculate a protocol $\pi$ over a noisy channel. In order to do so, we use the method
of Schulman~\cite{schulman96} described in Figure~\ref{alg:protocol}.
[The protocol is described for user $A$. The protocol for $B$ is identical.]
The idea behind the simulation is the following.
Each user keeps a record of  (his belief of) the  current progress of $\pi$,
described as a pebble on one of the \gametree nodes.

\begin{figure}[htb]
\begin{framed}
\small
Begin with own pebble at the root of \gametree\ and own state $S_A$ at the \statetree root's child labeled
$H \times \pi_A(x_A, \emptyset)$.
Repeat the following $N$ times\footnotemark:
\begin{enumerate}
\item Send $w(S_A)$ to user B.
\item Given the sequence of messages Z received so far from user B, guess the
current state $g$ of B as the node that minimizes $\Delta(W(g),Z)$.
From the guess $g$, infer
B's pebble movements and compute the (alleged) current position pebble$(g)$ of B's pebble and the bit $b$ outputted by B for this round.
\item Set your pebble movement and new state according to the current position $v$ of your pebble  and the following:
    \begin{enumerate}
    \item If $v={}$pebble$(g)$ then move own pebble according to the pair of bits
    $(\pi_A(x_A,v),b)$ to a state $v'$.
    The new state is $S_A$'s child labeled with the arc $(\pi_A(x_A,v),b)\times \pi_A( x_A, v')$.
    \item If $v$ is a strict ancestor of pebble$(g)$: own movement is $H$, and the
    next state is along the arc $H\times \pi_A(x_A,v)$.
    \item Otherwise, move own pebble backwards. New state is along the arc $R\times \pi_A(x_A,v')$
    where $v'$ is the parent of $v$.
    \end{enumerate}
\end{enumerate}
\end{framed}
\caption{Interactive protocol $\Sim$ for noisy channels~\cite{schulman96}}
\label{alg:protocol}
\end{figure}
\footnotetext{For the simulation to be well defined,
    we must extend $\pi$ to $N$ rounds.
    We assume that in each of the  $N-T$ spare rounds,
    $\pi$ outputs 0 for each user and every input.}

Each round, according to the transmissions received so far,
the user makes a guess for the position of the other user's pebble,
and infers how his own pebble should  move.
The user sends a message that describes
how he moves his pebble (out of the six possible movements matching the
4 child nodes, `H' to keep the pebble in the same place or `B' to back up to the parent node) and
the bit outputted by him, assuming the protocol is
described by the new position of his pebble.
Each one of these  12 options represents a child in a 12-ary
tree denoted as the \statetree (Figure~\ref{fig:trees}).
The user communicates\footnote
{We imply here using a (standard) error-correcting code in order to send the label over the
noisy channel, with constant slowdown (as given by Lemma~\ref{lem:shannon}).
Throughout the paper, any transmission of a label is to be understood in this manner.}
the label
assigned to the edge in the \statetree that describes his move. The \emph{state} of the user
is the current node on the \statetree, starting from its root,
and changing according to the edge communicated.

Informally speaking, the simulation works since the least common ancestor  of
both the user's pebbles always lie along the path $\gamma_{\pi,x}$. If both users take the
correct guess for the other user's pebble position, they simulate $\pi$ correctly and their pebbles move along $\gamma_{\pi,x}$.
Otherwise, their pebbles diverge, yet the common ancestor remains on $\gamma_{\pi,x}$.
On the following rounds, when the users acknowledge an inconsistency in the pebbles' positions,
they move their pebbles backwards until the pebbles reach their common ancestor, and
the protocol continues.  The users will simulate $\pi$ as long as the number of
divergences is small enough.
It is shown in~\cite{schulman96} that repeating the above process for
$N=O(T)$ rounds is sufficient for simulating
$\pi$ with exponentially small error probability (over the
channel errors).
We refer the reader to~\cite{schulman96}
for a detailed description of the protocol and its analysis.

In order to be able to construct such
a simulation, we replace the (non-constructive)
good 12-ary tree code with $\alpha=0.5$ originally used by Schulman, by a
\potent tree.
Surprisingly, this simple change is enough to obtain a constructible scheme
for simulating interactive protocols over noisy channels.
We note that in Section~\ref{sec:multiparty} we show that the same can be done using a ternary tree
instead of a 12-ary tree, using the methods of~\cite{RS94}. However, for the clarity of the
presentation, we first analyze the scheme using  a 12-ary tree
and only later optimize the result.

\begin{theorem}\label{thm:mainA}
Given a $(\frac{1}{10},\alpha)$-\potent tree code with a constant-size alphabet $|S|$
(which depends on the constant $ \alpha\in (0,1)$)
and an oracle for a decoding procedure of that tree code,
the protocol $\Sim$ (Figure~\ref{alg:protocol})
is an efficient simulation of the protocol $\pi$ (that has $T$ rounds).
It takes $N=O(T)$ rounds and succeeds with probability $2^{-\Omega(T)}$ over the channel errors,
assuming the use of an error correcting code with (label) error probability $p<2^{-42/\alpha}$.
\end{theorem}

Moreover,
if we are given an oracle to a tree code decoding procedure, the obtained protocol
is efficient. In Section~\ref{sec:eff} we show a decoding procedure
that is efficient  on average, given that the tree is \KTC. This immediately
leads to the following (main) Theorem.

\begin{theorem}[Main]\label{thm:main}
There exists an efficient simulation that computes any distributed 2-party protocol $\pi$ of length $T$,
using a BSC for communication and a pre-shared \KTC. The simulation imposes a constant slowdown,
and succeeds with probability $1-2^{-\Omega(T)}$ over the channel errors and the choice of the \KTC.
\end{theorem}

We now give the proof idea  for Theorem~\ref{thm:mainA} and later complete the formal proof.
We begin by defining a good move: a move that advances the simulation of $\pi$ in one step,
and a bad move: an erroneous step in the simulation that requires us to back up and re-simulate
that step. We show that any bad move is associated with a decoding error,
i.e., recovering a wrong node $u$, due to channel errors or tree defects.
Thus, we can bound the number of bad moves by bounding the probability for
channel errors and tree defects.
Using (Shannon's) error correcting codes, the probability of a channel error is arbitrarily small,
and so is the probability of having many channel errors.
Furthermore, we use a \potent tree code, to guarantee
a small number of tree defects, except with an exponentially small probability.

\vspace{1em}

%
Recall the following properties of the simulation $\Sim$.
\begin{lemma}[\cite{schulman96}] \label{lem:lcaOnTrack}
The least common ancestor of the two pebbles lies on $\gamma_{\pi,x}$.
\end{lemma}
\begin{lemma}[\cite{schulman96}]\label{lem:moves}
Let $v_A$ and $v_B$ be the positions of the two pebbles in the \gametree at some time $t$,
and let $\bar v$ denote the least common ancestor of $v_A$ and $v_B$.
Define the {\em mark} of the protocol as the depth of $\bar v$ minus the
distance from $\bar v$ to the further of $v_A$ and $v_B$.

If during a specific round, both users guess the other's state correctly
(a {\em good} move), the mark increases by 1.
Otherwise (a {\em bad} move), the mark decreases by at most 3.
\end{lemma}
\noindent A proof for  both of the above lemmas is given in~\cite{schulman96}.

Our goal is to show that the probability of having more than $cN$ bad rounds
is exponentially small. By setting $c=1/5$ and $N=5T$ we guarantee
that  at the end of the calculation the mark will be (at least) $T$.
Since the common ancestor of the pebbles always lies along the path $\gamma_{\pi,x}$,
a mark of value $T$  indicates that the common ancestor has reached depth $T$,
and $\pi$ was successfully simulated.

For a bad round at time $t$, we assume that (at least) one of the users takes a wrong
guess of the (other user's) current state.
Suppose that
the least common ancestor of the right state and
the wrongly guessed state in the \statetree, is distanced $l$ levels away (i.e., an error of magnitude $l$).
Define the {\em error interval} (of length $l$) corresponding to the erroneous guess
as  $[t-l, t]$.

We now show that given a \potent tree,
$\Sim$  simulates $\pi$ over a noisy channel with overwhelming probability.

\begin{proof} \textbf{(Theorem~\ref{thm:mainA})}.
Suppose the parties share a $(\frac{1}{10},\alpha)$-\potent tree code\footnote{
    Proposition~\ref{prop:KTCisPotent} guarantees  that
    as long as $|S|>(2d)^{22/(1-\alpha)}$,
    only a negligible fraction of the \KTC{}s are $(\frac{1}{10},\alpha)$-bad. Therefore,
     for obtaining a \potent tree with overwhelming probability, we require $\log |S|\ge101$.},
for some $0<\alpha<1$.
Assume that a specific run of a simulation failed, and thus
it must be that more than $N/5$ errors have occurred.

Note that the simulation defines a path along the \statetree, from the root
to one of the leaves. Since the \statetree is $(\frac{1}{10},\alpha)$-\potent, the
specific path contains bad intervals of total length
at most $N/10$. We assume a worst case scenario in which
each $\alpha$-bad node causes a bad move in the simulation.
We show that
the probability of having $N/10$ additional  bad moves
(in the remaining nodes, which are not $\alpha$-bad)
is exponentially small.

Consider a specific bad move caused by erroneously decoding a node which is not $\alpha$-bad,
at time $t$. Namely, the user guesses a wrong node $r$ instead of the real transmitted node $s$.
For an error of magnitude $l_i$, $W(s)$ and $W(r)$ are identical from the root
to the least common ancestor of $r$ and $s$ at level $t-l$.
Since the decoding is done by minimizing the Hamming distance,
making this wrong guess is independent of transmissions prior to round $t-l$.
It follows that such an error (of magnitude $l$) can happen only if at least $\alpha l/2$
channel errors have occurred during the last $l$ rounds.
Due to the same reason, it is easy to see that decoding errors of which the
error intervals are disjoint, are independent.

We consider again the bad moves which are associated with
a decoding error of  nodes which are not $\alpha$-bad. Each such a bad move (i.e., a decoding error)
impose an error interval of length $l_i >1$, and the union of these intervals must be
of length at least $N/10$.
Each such an error happens with probability at most
$\sum_{j=\alpha l_i/2}^{l_i} {{l_i} \choose j} p^j \le 2^{l_i} p^{\alpha l_i/2}$.
Due to Lemma~\ref{lem:intervals} we can find a set of disjoint intervals
of length at least $N/20$. Due to the discussion above, these errors are independent, and their
probability to jointly occur is bounded by
\begin{align*}
 \prod_{i} 2^{l_i} p^{\alpha l_i/2} = (2p^{\alpha/2})^{l}\text{.}
\end{align*}

We conclude the proof by bounding the probability
for having any possible error pattern of total length at least $N/20$
along the bad moves associated with nodes which are not $\alpha$-bad,
by using the union bound
over all possible error patterns (there are
at most $\sum_{j=N/20}^N{N \choose j}\le 2^N$ such patterns),
for each one of the users. The probability is bounded by
$$
\sum_{\text{user } U}\ \sum_{\stackrel{\text{\scriptsize pattern of }}{l\ge N/20 \text{ errors}} }  (2p^{\alpha/2})^{l} \le
2\cdot 2^N (2p^{\alpha/2})^{N/20},
$$
which is $2^{-\Omega (N)}=2^{-\Omega (T)}$ for  $p< 2^{-42/\alpha}$.
\end{proof}

\subsubsection{Performing decoding in an efficient way}\label{sec:eff}

A decoding process outputs the node $u$ (at depth $t$) that minimizes the Hamming distance
between $W(u)$ and the received string of labels ${\mathbf{r}=r_1r_2\dotsm r_t}$.
Although the above Theorem~\ref{thm:mainA} is proven assuming an oracle to
tree-code decoding procedure, this requirement is too strong for our needs.
Since we count any node which is $\alpha$-bad as an error (even when no error have occurred),
it suffices to have an oracle that decodes correctly given that the (transmitted) node
is not $\alpha$-bad.

We follow the techniques
employed by Schulman~\cite{schulman96}
(which are based on ideas from~\cite{wozencraft57, reiffen60, fano63}),
and show an efficient decoding that succeeds if the node is
not $\alpha$-bad.
While the decoding process of~\cite{schulman96} is based on the fact that
the underlying tree is a good tree code,
in our case the tree code is a \KTC.\footnote{A similar proof works also for a \RTBC.}

The decoding procedure is the following.
For a fixed time $t$, let $g_{t-1}$ be the current guess of the other user's state,
and denote the node along the path from the root to $g_{t-1}$ as $g_1, g_2, \ldots, g_{t-1}$.
Also, recall that  $r_1,r_2, \ldots, r_t$ are the labels received so far. If there exists a child of $g_{t-1}$
whose edge is labeled $r_t$, choose that child (break ties arbitrarily),
otherwise, arbitrarily choose one of $g_{t-1}$'s child nodes.
Denote with $g_t$ the new guess.

Recall that $W_m(u)$ denotes the $m$-suffix of $W(u)$, i.e., the last $m$ symbols along the path from the tree's root to the node $u$.
We look at the earliest time $i$ such that $\Delta( r_ir_{i+1}\cdots r_t, W_{t-i+1}(g_t)) \ge \alpha (t-i)/2$.
For that specific $i$, exhaustively search the subtree of $g_i$ and set
the new guess $g$ as the node $u$
(at depth $t$) that minimizes the Hamming distance $\Delta(r_1r_2\cdots r_t, W(u))$.

Note that when $g_t$ is an $\alpha$-bad node of maximal length $l$,
any path from the root to some other node $g'_t$,
where the least common ancestor of $g_t$ and $g'_t$ is located $l'>l$ levels away,
must have
a Hamming distance $\Delta(W_{l'}(g_t) , W_{l'}(g'_t)) \ge \alpha l'$.
Therefore, if all the suffixes of length $l'>l$ satisfy
$\Delta( r_{t-l'+1}\cdots r_t, W_{l'}(g_t)) < \alpha l'/2$,
it is guaranteed that the
node minimizing the Hamming distance is within the subtree of $g_{t-l}$.
However, if $g_t$ is an $\alpha$-bad node of length $l$,
the decoding process might yield a wrong guess,
i.e., a node in the subtree of $g_{t-l}$ that does not minimize the Hamming distance.

The following proposition bounds the probability for a decoding error of magnitude $l$.
\begin{proposition}\label{pro:errorProb}
Assume a \KTC is used to communicate the string $W(v)$ over a BSC.
Using the efficient decoding procedure (with some constant $\alpha\in(0,1)$),
the probability for a specific user to make
a decoding error of magnitude $l$
is  bounded by $2\left(\frac{4d}{|S|}\right)^l + 2\left(\frac{2d}{|S|^{1-\alpha}}\right)^l$,
if an error correction code with (label)
error probability less than $|S|^{-2}$ is used.
\end{proposition}

\begin{proof}
A decoding error of magnitude $l$ occurs if the decoding process outputs a node
 $u\ne v$, such that the common ancestor of $u,v$ is $l$ levels away.
Such an error can happen due to one of the following reasons:
\begin{enumerate}
\renewcommand{\labelenumi}{(\roman{enumi})}
\setlength{\itemsep}{1pt}
\vspace{-0.3em}
\item For the received string ${\mathbf r}=r_1r_2 \ldots r_l$ it holds that
$\Delta({\mathbf r}, W(u)) \le \Delta({\mathbf r}, W(v))$. This happens when
the Hamming distance $\Delta(W(u),W(v))$ is $j=0,1,\ldots,l$ and more than $j/2$ channel errors  occurred.
\item The decoding process did not return the node that minimizes the Hamming distance.
\end{enumerate}
Note that we only need to consider the paths from root to $u$ and to $v$ and thus use the $2N$-wise independence of the tree's labels.
Recall that the probability to have specific set of $l<2N$ labels is $2^{-cN\log|S|}$ away from uniform, with $c=O(1)$, and the probability for a given Haming distance between $W(u)$ and $W(v)$ is bounded by Lemma~\ref{lem:ktcHD}.
Let $p < |S|^{-2}$ be the maximal label error of the channel.
Using a union bound for every possible node $u$,
the probability of part (i) is bounded by
\begin{align*}
\Pr[ \text{ Error of}&\text{ magnitude } l\ ]  \\ &
\le  \sum_u\sum_{j=0}^l \Pr [\Delta (W(v),W(u)) = j  ] \cdot  \Pr [\ge j/2 \text{ symbol-errors}] \\
& \le d^l \sum_{j=0}^{l}  2{l \choose l-j} \left(\frac{1}{|S|}\right)^{l-j}
   \sum_{k=j/2}^l {l \choose k} p^k (1-p)^{l-k} \\
& \le 2\cdot d^l \sum_{j=0}^{l}  {l \choose l-j} \sum_{k=j/2}^l {l \choose k} {|S|}^{j-l} |S|^{-2k} \\ &
  \le 2\cdot d^l  \cdot 2^l\cdot 2^l\cdot |S|^{-l}\text{ ,}
\end{align*}
which is exponentially small in $l$ as long as $|S|>4d$.

For part (ii),
note that the decoding process does not return the node that minimizes the Hamming
distance if
$l$ is larger than $t-i$, for the suffix determined by the decoding procedure
(using the notations described above for the efficient decoding procedure).
This implies that for the outputted node $g_t$, $\Delta(r_i\cdots r_t,W_i^{t}(g_t)) < \alpha(t-i)/2$.
Since $g_t$ is not the node that minimizes the Hamming distance,
there must exist a node $v$
of distance at most $l$,
such that $\Delta(W_i^t(v), r_i\cdots r_t) \le  \Delta(r_i\cdots r_t,W_i^{t}(g_t))$.
By the triangle inequality, the Hamming distance between the paths from
 $v$ and $g_t$  to
 their least common ancestor
must be at most $\alpha l$.
Using the union bound for any possible such $v$ and any possible Hamming distance up to $\alpha l$,
we bound the probability of this event by
$$d^l\sum_{j=0}^{\alpha l}2{l \choose l-j}|S|^{-(l-j)}\le 2(2d)^l |S|^{-l(1-\alpha)}\text{ .}$$
A union bound on the two cases completes this proof.
\end{proof}

We stress that the above decoding process always outputs the correct
node (i.e., the node which minimizes the Hamming distance), if the transmitted
node is not $\alpha$-bad. For that reason, the analysis performed
in the proof of  Theorem~\ref{thm:mainA} is still valid, since it
only requires the decoding procedure to succeed when the node is not $\alpha$-bad
(and assumes that the simulation has a bad move
in each node which is a bad node).

We now show that this procedure is efficient in expectation.
Let $L(t)$ be the depth of the subtree explored at time $t$.
The decoding process takes $O\big(\sum_{t=1}^N d^{L(t)}\big)$  steps
(this dominates terms of  $O(L(t))$ required to maintain the guess, etc).

For time $t$, if $L(t)=l$ then $\Delta(r_{t-l+1}\cdots r_t, W_{l}(g_t)) \ge \alpha l/2$
yet for $l'>l$, $\Delta(r_{t-l'+1}\cdots r_t, W_{l'}(g_t)) < \alpha l'/2$. Assume that the
sequence of labels transmitted is $W(v)$ for some node $v$ of depth $t$.
The above requirements imply that the suffixes (of length $l$)
of $W(v)$ and $W(g_t)$ have Hamming distances \emph{exactly} $\lceil \alpha l/2\rceil$.
This happens with probability at most
\begin{align*}
\le &\sum_{j=0}^{\lceil\alpha l/2\rceil-1}\Pr[\Delta(W_{l}(g_t),W_{l}(v) = j]\Pr [\ge \alpha l/2-j \text{ symbol-errors}] \\
&+ \sum_{j=\lceil\alpha l/2\rceil}^l \Pr[\Delta(W_{l}(g_t),W_{l}(v) = j]\Pr [\ge j-\alpha l/2 \text{ symbol-errors}]
\\
\le & \sum_{j=0}^l 2{l \choose l-j} \left(\frac{1}{|S|}\right)^{l-j}
\sum_{k=|\alpha l/2-j|}^l{l \choose k} p^{k} (1-p)^{l-k}
\le  2^{2l+1} |S|^{-l(1-\alpha/2)}\text{ ,}
\end{align*}
assuming $p<|S|^{-2}$.

With a sufficiently large yet constant alphabet, e.g.,  $|S|>(8d)^{1/(1-\alpha/2)}$,
we bound
the probability that $L(t)$ equals $l$ to be $2^{-\gamma l}< d^{-l}$.
The expected running time is then given by
\[
O\bigg(  \sum_{t=1}^N E \Big[ d^{L(t)}\Big]\bigg) =
O\bigg(  \sum_{t=1}^N\sum_{l=0}^t \Big[ 2^{-\gamma l}d^{l}\Big]\bigg) =
O\bigg(  \sum_{t=1}^N \frac{2^\gamma}{2^{\gamma} -d} \bigg) = O(N)\text{.}
\]
Since we repeat the simulation step for  $N=O(T)$ times,
the computation is efficient in expectation.
To complete the proof, we mention that~\cite{schulman96} presents  a data structure
which allows us to perform the above decoding with overhead $O(L(t))$.

\subsubsection{Simulating an adaptively chosen protocol}\label{sec:adaptive}

For a given protocol,  $\Sim$ fails with exponentially small probability that
depends on the choice of the \KTC and the BSC errors.
Assume that the we first pick a \potent tree
and then the protocol $\pi$ is (adversarially) chosen.
Due to Theorem~\ref{thm:mainA}, as long as the tree code is  \potent,
the simulation succeeds with overwhelming probability, over the
BSC errors alone. However, the decoding process described in
Section~\ref{sec:eff} above, might no longer be efficient,
since the adversary might force the simulation
to travel through the ``bad'' regions in the tree that require exploring large subtrees.

An interesting remedy to the above
can be achieved by
by introducing more randomness,
which prevents the adversary from fixing the path along the \statetree the simulation takes.
We now extend  the basic scheme $\Sim$
to the stronger notion
of adversarially chosen protocol (Section~\ref{sec:adaptive}), and prove the following theorem.
\begin{theorem}\label{thm:mainAdaptive}
Except for probability $2^{-\Omega(T)}$ over the choice of the \KTC, there exists
an efficient scheme to simulate any 2-party protocol $\pi$
of length $T$, with success probability at least $1-2^{-\Omega(T)}$ over the channel errors.
\end{theorem}

As said above,
the expected runtime for the decoding process described in Section~\ref{sec:eff}
is no longer efficient in this case. The adversary can choose
the simulated protocol $\pi$ and  ``fix'' a path along the
\statetree (up to channel errors). As the decoding process is efficient in expectation,
the path fixed by the adversary might be a path that is inefficient to decode.
However, by adding randomness, we are able to change $\Sim$ such that
for any protocol $\pi$, the actual traversed path in the \statetree is fully random.
This is done by permuting the nodes of the \statetree
separately for each level. The users need to communicate which permutation
is used for each level, which is done by sending the specific permutation in use via
additional (\potent) tree code, which we denoted as the \randomtree.

Define the \randomtree to be a \KTC of degree $d!$ (for our case $d=12$).
For a specific node, each one of  the $d!$ children
denotes one of the possible permutations on $d$ values.
Each round, the user chooses a random permutation by randomly selecting
one of the children of his current position in the \randomtree (starting from the root).
Recall that, in the \statetree, each node has 12 children where each represents one of
$\{00\text{x}0$, $00\text{x}1$, $\ldots\}$. We can assume a fixed order, that is, the first child always
represents $00\text{x}0$, the second represents $00\text{x}1$, etc.
For a time $t$, assume the chosen permutation is $P_t$.
In our randomized simulation, the $i^{\text{th}}$ child in the
\statetree has the meaning  $P_t(i)$. For instance, the first node represents
one of the meanings $\{00\text{x}0$, $00\text{x}1$, $\ldots\}$, determined by $P_t(1)$.

The adapted scheme, Randomized-$\Sim$, is described in Figure~\ref{alg:protocolAdaptive}.
For Theorem~\ref{thm:mainAdaptive}, we assume that
both the \statetree and the \randomtree are $(\frac1{20},\alpha)$-\potent trees,
which can be achieved (with overwhelming probability)
by having a large enough (yet a constant) alphabet size.
As before, we assume each label is sent using an error correcting code, such that
the error probability per transmission is less then $\min (|S|^{-2}, 2^{-82/\alpha})$.
Such an error correcting rate imposes
a constant slowdown, according to Lemma~\ref{lem:shannon}.

\begin{figure}[htb]
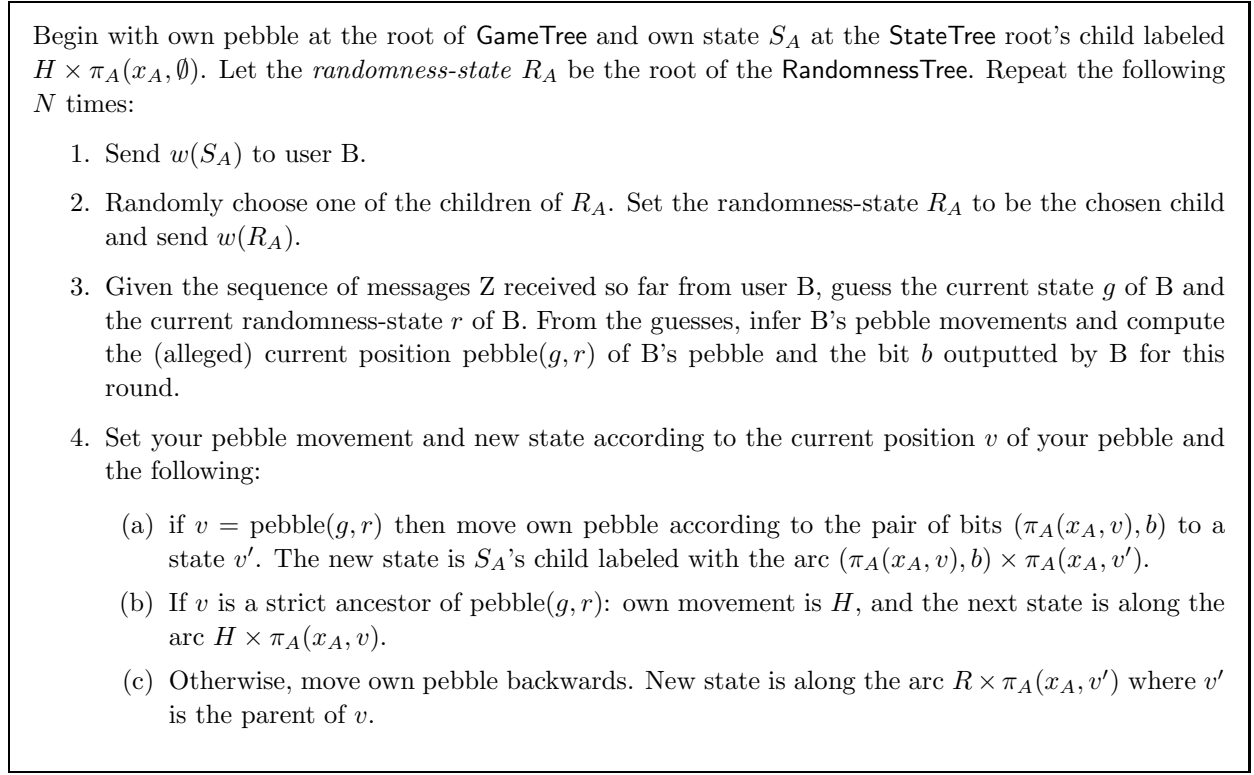

\begin{framed}
\small
Begin with own pebble at the root of \gametree\ and own state $S_A$ at the \statetree root's child labeled
$H \times \pi_A(x_A, \emptyset)$. Let the \emph{randomness-state} $R_A$ be the root of the \randomtree.
Repeat the following $N$ times:
\begin{enumerate}
\item Send $w(S_A)$ to user B.
\item Randomly choose one of the children of $R_A$. Set the randomness-state $R_A$ to be
the chosen child and send $w(R_A)$.
\item Given the sequence of messages Z received so far from user B, guess the
current state $g$ of B and the current randomness-state $r$ of B.
From the guesses, infer B's pebble movements and compute the
(alleged) current position pebble$(g,r)$ of B's pebble and the bit $b$ outputted by B for this round.
\item Set your pebble movement and new state
according to the current position $v$ of your pebble  and the following:
    \begin{enumerate}
    \item if $v={}$pebble$(g,r)$ then move own pebble according to the pair of bits
    $(\pi_A(x_A,v),b)$ to a state $v'$.
    The new state is $S_A$'s child labeled with the arc $(\pi_A(x_A,v),b)\times \pi_A( x_A, v')$.
    \item If $v$ is a strict ancestor of pebble$(g,r)$: own movement is $H$, and the
    next state is along the arc $H\times \pi_A(x_A,v)$.
    \item Otherwise, move own pebble backwards. New state is along the arc $R\times \pi_A(x_A,v')$
    where $v'$ is the parent of $v$.
    \end{enumerate}
\end{enumerate}
\end{framed}
\caption{Interactive protocol Randomized-$\Sim$ for
  simulating an adversarially chosen protocol $\pi$ over noisy channels}
\label{alg:protocolAdaptive}
\end{figure}

\vspace{0.5em plus 0.5em minus 0.2em}
Theorem~\ref{thm:mainAdaptive} immediately follows from the following theorems,
\begin{theorem}\label{thm:mainB}
Assume \randomtree and \statetree are $(\frac{1}{20},\alpha)$-\potent
trees for some $\alpha\in(0,1)$, and assume an oracle for tree-code decoding process,
then Randomized-$\Sim$ (Figure~\ref{alg:protocolAdaptive}) is an efficient simulation
of any protocol $\pi$ (that has $T$ rounds).
If any label is sent using a error correcting code with (label) error probability
$p$, then Randomized-$\Sim$
succeeds except with probability
$4\cdot2^N(2p^{\alpha/2})^{N/40}$ over the channel errors.
\end{theorem}
\begin{theorem}\label{thm:adaptiveEff}
For a given $\alpha\in(0,1)$, suppose that
the \randomtree and the \statetree are \RTBC{s} with
large enough (yet constant) alphabet size,  
and that $p < |S|^{-2}$,
then the decoding procedure (Section~\ref{sec:eff}) is
(i) efficient in expectation, and (ii) correctly decodes a node which is not $\alpha$-bad.
\end{theorem}

We now show that the simulation succeeds except with an exponentially small probability.
Note that when both the \statetree and \randomtree are correctly decoded,
the user recovers the position of the other user's pebble and the move
is successful (in the notion of Lemma~\ref{lem:moves}).
However, if for a time $t$ there is an error either in the \statetree or in the \randomtree,
the move is a bad move.

\begin{proof}\textbf{(Theorem~\ref{thm:mainB})}.
Assume that a simulation failed, which means more than $N/5$ bad moves have occurred.
We fix the path traveled in both trees,
and show the probability of having this many errors is exponentially small.
Since our trees are $(\frac1{20},\alpha)$-\potent, each path includes bad intervals of total length
at most $N/20$, which can contribute
towards at most $N/10$ bad moves, for both trees.

The rest of the errors have occurred in nodes which are not $\alpha$-bad, and for at least
one of the trees, the number of such errors is at least $N/20$. Consider that specific tree.
Each such an error is associated with error interval of length $l_i>1$, such that
the length of the union of the intervals is at least $N/20$.
Using Lemma~\ref{lem:intervals}, we know there exist disjoint intervals of total length
at least $N/40$ along the same path.
Each error interval of length $l_i$ corresponds to an error of magnitude $l_i$, and
since the transmitted node is not $\alpha$-bad, these errors can only be caused by channel errors.
As explained above (see proof of Theorem~\ref{thm:mainA}),
the errors which correspond to these disjoint intervals are independent,
and the probability that a fixed specific pattern of disjoint intervals
of total length $l$ jointly occur
is at most $(2p^{\alpha/2})^{l}$, where $p$ is the label error probability.

The probability
of having any error pattern of total length at least $N/40$
is given by the union bound, summing over
all possible error patterns in both tree, and over each one of the users. There are
at most $2\cdot 2\sum_{j=N/40}^N{N \choose j}\le 4\cdot2^N$ such combinations, and
the probability is bounded by $4\cdot2^N(2p^{\alpha/2})^{N/40}$, which is exponentially small for
$p < 2^{-82/\alpha}$.
\end{proof}

We now show that the decoding procedure described in Section~\ref{sec:eff} is efficient in expectation, for both the trees.
\begin{proof} \textbf{(Theorem~\ref{thm:adaptiveEff})}.
The Randomized-$\Sim$ scheme makes a random walk\footnote{
    To be more accurate,
    it is a random walk from the root to a leaf, where the depth can only increase.}
on the \randomtree.
Following the analysis of Section~\ref{sec:eff},
we can choose a constant-size $S$ such that
the probability of exploring at time $t$ a subtree of depth $L(t)=l$
is bounded by $2^{-\gamma l} < (d!)^{-l}$
(when decoding a \randomtree node,
using the decoding procedure described in Section~\ref{sec:eff}).
Specifically, let
$|S| > (8d!)^{1/(1-\alpha/2)}$,
and require $p<|S|^{-2}$.
The expected time for decoding
the \randomtree during the simulation is given by
\[
O\bigg(  \sum_{t=1}^N E \Big[ (d!)^{L(t)}\Big]\bigg) =
O\bigg(  \sum_{t=1}^N\sum_{l=0}^t \Big[ 2^{-\gamma l}(d!)^{l}\Big]\bigg) =
O\bigg(  \sum_{t=1}^N \frac{2^\gamma}{2^{\gamma} -(d!)} \bigg) = O(N)\text{.}
\]

For a fixed protocol $\pi$, each path of the \randomtree
defines a corresponding path in the \statetree. Since the \randomtree contains all
possible permutations
for the \statetree nodes, choosing a random path in the \randomtree yields a random walk on the
\statetree. In the same manner as above,
the expected time for decoding the \statetree is efficient.

Property (ii) has been proven in Section~\ref{sec:eff}, and holds for this case as well.
\end{proof}

\subsection{Simulating $n$-Party Protocols}\label{sec:multiparty}

In this section we extend our result to support a simulation
of a protocol $\pi$ with any number $n$ of users.
This is done by incorporating the tools described in the previous sections
with the method of simulating an $n$-party protocol over a disturbed channel
developed by Rajagopalan and Schulman~\cite{RS94}.
The paper~\cite{RS94} shows that a scheme for simulating multiparty protocol
over a disturbed channel
{\em exists},
yet the question of its efficient implementation has been open since 1994.
The Scheme presented in~\cite{RS94}
obtains a communication dilation of $O(\log r)$ where $r$ is the maximal connectivity degree,
that is, the maximal number of parties connected to a specific user.

Rajagopalan and Schulman, in their work~\cite{RS94},
describe how to adapt the 2-party simulation of~\cite{schulman96}
to an arbitrary number of users. The key idea is to replace the 12-ary \statetree
with a ternary tree (that is, $d=3$),
where each node has three child nodes marked with $\{0,1,bkp\}$.
The values $0$ and $1$ indicate the output bit of the user in the simulated round,
and $bkp$ indicates that the last simulated round is suspected to be invalid
and should be  deleted and re-simulated.
The simulation (described here for a specific user $i$)
is completely defined by the following process. Each round, the user uses all the previous communications to infer the current simulated round of $\pi$ and sends his output bit to user $j$ (by communicating the label assigned with arc to child  $0$ or $1$ respectively,
in the ternary \statetree shared between users $i$ and $j$).
If the user finds an inconsistency, he transmits
$bkp$ which denotes deleting the last received (undeleted) bit and rolling the protocol $\pi$ one step back. The user shares such a ternary tree with each of the $r$ parties connected to him, and is allowed
to output a different bit to each party. Yet, when the user decides to roll back he outputs
$bkp$ on each of the outgoing links. Inconsistency is defined as one of the two following cases:
(1) the current decoded transcript of the \statetree disagrees with the bits sent so far, or (2)
the user received $bkp$ from one of his neighbors.
We refer the reader to~\cite{RS94} for a complete description and analysis of this scheme.

One can easily check that the bulk of the analysis performed in~\cite{RS94} applies
for the case of replacing the good ternary tree code with a ternary \KTC (or \RTBC).
The analysis is composed of two parts. The first part shows that if after
$t$ rounds the scheme simulates step $t-l$ of $\pi$ then at least
$l/2$ errors have occurred in decoding the correct tree-node during
the \emph{history cone} of the user at time $t$ (i.e., all the transmissions
that affect the user state at time $t$). The other part
bounds the probability of having a constant fraction of errors (out of the number of rounds).
While the first part is completely independent of the
fact that we replace the good tree code with a \KTC,
in order to complete the proof, we must adapt the second part to the usage of \KTC. This is done by
Lemma~\ref{lem:raja} below.

Let us formally describe these two parts.
We begin by defining the notion of the history cone~\cite{RS94}.
Let $(p,t)$ denote a user $p$ at time $t$.
\begin{definition}[\cite{RS94}]
$(p,t)$ and $(p',t')$ are \textbf{time-like} if messages sent by user $p$ at time $t$ has an affect on
the computation of user $p'$ at time $t'$ (or vice versa).
\end{definition}
That is, $(p,t)$ and $(p,t')$ are always time-like,
and $(p,t)$ and $(q,t+1)$ are time-like if $p$ and $q$ are neighbors.
\begin{definition}[\cite{RS94}]
A \textbf{$t$ time-like path} is a sequence $\{ (p_i,i)\ |\ 1\le i\le t\}$
such that  any two elements in the path are time-like
(i.e., for every $i$, $p_i$ and $p_{i+1}$ are either neighbors or the same party).
\end{definition}

The proof of~\cite{RS94} follows from the next two lemmas
\begin{lemma}[Lemma (5.1.1) of~\cite{RS94}]\label{lem:rsOne}
If a user $p$ at time $t$ has successfully simulated only the first $t-l$ rounds of $\pi$,
then there is a $t$ time-like sequence that ends at $(p,t)$ and
includes at least $l/2$ tree-decoding errors.
\end{lemma}
\begin{lemma}[Lemma (5.1.2) of~\cite{RS94}]\label{lem:rsTwo}
Using error correcting codes with dilation $O(\log(r+1))$,
the probability that any fixed t time-like path has more than
$t/4$ tree-decoding errors, is less than
$\frac{1}{(2(r+1))^t}$.
\end{lemma}
The proof of the multiparty case is given by setting $t=N=2T$.
The first lemma states that if the simulation failed
(the first $N/2$ rounds of $\pi$ are not valid for some user)
then there must exist one user who has $N/4$ errors along one of his $N$ time-like sequences.
The probability of this event is bounded by the Lemma~\ref{lem:rsTwo} to be
less than $\frac{1}{(2(r+1))^t}$ summed over
all the $N(r+1)^N$ possible time sequences, which is bounded by $N2^{-N}$.

While the above Lemma~\ref{lem:rsOne} holds regardless of the tree in use,
we prove a variant of the above Lemma~\ref{lem:rsTwo}
for the case of using a \potent tree.
Moreover, although Lemma~\ref{lem:rsTwo} holds for any time $1\le t\le N$, only $t=N$ is required
for completing the proof for the multiparty case, which we prove in the following lemma.

\begin{lemma}\label{lem:raja}
Suppose each two users share a $(\frac{1}{16n},\alpha)$-\potent tree, for some
 $\alpha\in(0,1)$. If an error correcting code with  label error probability $p$ is used,
then for any fixed $N$ time-like path, the probability that there are more than
$N/4$ tree-decoding errors
is bounded by $(2^{17}p^{\alpha/2})^{N/16}$,
over the errors of the channel.
\end{lemma}
\begin{proof}
We assume an oracle for the decoding process, which can easily be replaced
by the \emph{efficient decoding} procedure
given in Section~\ref{sec:eff}, if we use a \KTC.
Assume that at least $N/4$ errors have occurred in a specific $N$ time-like path.
Fix a specific user $i$ and assume that the errors of this user
are included in error intervals of total length $l_i$.
By Lemma~\ref{lem:intervals}, there exist disjoint error intervals of total length at least $l_i/2$.
Recall that each error interval of length $\ell$ corresponds to an error of magnitude $\ell$,
and recall that in each tree, at most $N/16n$ of the nodes are $\alpha$-bad.
Thus, at least $k_i\equiv\max\{0, l_i/2-N/16n\}$ of the errors of user $i$ in the $N$ time-like path
occur in nodes which are not $\alpha$-bad.
These errors can only be originated due to channel errors\footnote{This claim also applies
    to the efficient decoding procedure, as it
    always returns the node that minimized the Hamming distance,
    if it is not $\alpha$-bad.
    See the proof of Theorem~\ref{thm:mainA}
    and discussion in Section~\ref{sec:eff}.},
and since the intervals are disjoint, they are independent.
As above (see proof of Theorem~\ref{thm:mainA}),
the probability of having errors that correspond to these (fixed) disjoint error intervals
is bounded by $2^{k_i}p^{\alpha k_i/2}$.
Clearly, tree-decoding errors of a specific user are independent
of the communication (and channel errors) of other users.
It follows that the probability for all the users to have
a total amount of $N/4$
errors matching the fixed intervals pattern is bounded by
$(2p^{\alpha/2})^{\sum_{i}k_i}$.
With $\sum_i l_i >N/4$ and at most $n$ users,
this probability is bounded by $(2p^{\alpha/2})^{N/8-n(N/16n)}=(2p^{\alpha/2})^{N/16}$.

Using a union bound we sum the probability over any number $j\ge N/4$ of errors  and
over any one of the  ${N \choose j}$
different ways to distribute $j$ errors along the fixed time-like path.
The probability that there are at least $N/4$
errors in this fixed $N$ time-like path is bounded by
\[
\sum_{j=N/4}^N {N \choose j}(2p^{\alpha/2})^{j/2-N/16} \le  (2^{17}p^{\alpha/2})^{N/16}\text{ .}
\]
\end{proof}
\noindent For  $p<(5(r+1))^{-32/\alpha}$,
this probability is at most $\frac{1}{(2(r+1))^N}$.
\begin{corollary}
Suppose each two users\footnote{The same tree can be used by all the users.} share a \KTC
with $|S|\ge ((2d)^{32n+2})^{1/(1-\alpha)}$ for some $\alpha\in(0,1)$,
and use an error correcting code with
(label) error probability less than $p\le(5(r+1))^{-32/\alpha}$.
Then, except with probability $2^{-\Omega(N)}$ over the choice of the \KTC,
for any fixed $N$ time-like path,
the probability that there are more than $N/4$ tree-decoding errors
is less than $\frac{1}{(2(r+1))^N}$
over the the errors of the channel.
\end{corollary}
That is, with $|S|\ge ((2d)^{32n+2})^{1/(1-\alpha)}$ the \KTC is $(\frac1{16n},\alpha)$-\potent, with
overwhelming probability, due to Proposition~\ref{prop:KTCisPotent}.
Each label in an alphabet of size $|S|$ requires $\log |S|=O(n)$ bits.
Due to Lemma~\ref{lem:shannon}, we can use an error correcting code
such that each transmission is $O(n)$ and the label error probability is
less than the required $(5(r+1))^{-32/\alpha}$. Specifically,
for efficient decoding we require $p  < |S|^{-2}$,
which can be done with code of length
$O(n)$ as well.
The above lemma replaces
Lemma 5.1.2 of~\cite{RS94}, and
leads to  the following theorem.
\begin{theorem}\label{thm:multiparty}
There exists a constructible and efficient  simulation
that computes any $n$-party protocol $\pi$ of length $T$
using a BSC for communication and a pre-shared \KTC.
The simulation succeeds with probability $2^{-\Omega(T)}$,
and impose a dilation of $O(n)$.
\end{theorem}
\noindent An efficient version of the above scheme,
using the efficient decoding methods described in Section~\ref{sec:eff},
has an expected time complexity of $O(Tr)$.

\section*{Acknowledgments}
We would like to thank Leonard Schulman and Anant Sahai for many useful discussions at a very early stage of this research. We also thank Madhu Sudan, David Zuckerman, and Venkatesan Guruswami for several helpful conversations.  We would like to thank Alan Roytman for miscellaneous remarks.


\bibliographystyle{alpha}
\bibliography{coding}


%

\appendix
\section*{Appendix}

\section{Random Tree Codes And Their Properties}\label{app:RTBC}
In this section we analyze several of the properties of Random Tree Codes (\RTBC), and show that
a \RTBC is potent, except with high probability.

Observe that any two paths in a \RTBC have large enough Hamming distance, except for  a
negligible probability over the choice of the labels. This property makes the \RTBC a useful code.
\begin{lemma}\label{lem:prob4smallHamming}
Let $\tree$ be a $d$-ary \RTBC over $S$,
and let $v_1$ and $v_2$ be any two nodes at some common depth $h$ in $\tree$,
with least common ancestor  at depth $h-l$, then for every $0\le\alpha\le1$,
$
\Pr \left[\Delta (W(v_1),W(v_2)) \le \alpha l \right] \le \Big(\frac{2}{|S|^{1-\alpha}}\Big)^l
$
\end{lemma}
\begin{proof}
We sum the probability for any possible Hamming distance $i=0,1,\ldots,\alpha l$. A direct calculation gives\\
\[
\Pr [\Delta (W(v_1),W(v_2)) \le \alpha l ] \le \sum_{i=0}^{\alpha l}{l \choose l-i} \left(\frac{1}{|S|}\right)^{l-i}\left(\frac{|S|-1}{|S|}\right)^{i}
\le  \frac{2^l}{|S|^{l(1-\alpha)}} \text{ .}
\]
\end{proof}\vspace{-0.5em}
Assume that $v, u$ are at some depth $h$,
and that their least common ancestor is at depth $h-l$;
we say that $u$ and $v$ have a distance $l$ in that case.
Assume that the labels $W(v)$ were transmitted.
If a good tree code is used, then the probability of decoding a different node, $u$, is
exponentially small in the distance $l$, where the probability is
over the channel errors. We denote this event as a
decoding error of \emph{magnitude} $l$.
In the following lemma we obtain a similar
result for a random tree code,
where in this case the probability is over both the channel errors and the choice of the \RTBC.

Finally, we prove Theorem~\ref{thm:RTCisPotent} by showing that
the set of all $(\epsilon,\alpha)$-bad \RTBC
for constants $\epsilon,\alpha \in(0,1)$, is exponentially small.
\begin{proposition}\label{lem:probBadRTBC}
Suppose $\eps,\alpha \in (0,1)$.
The probability for a \RTBC of depth $N$
with alphabet $|S|>(2d)^{(1+2/\eps)/(1-\alpha)}$
to be $(\varepsilon,\alpha)$-bad,
is at most $2^{-\Omega(N)}$
\end{proposition}
\begin{proof}
We begin by fixing a leaf $z$, and later we use a union bound
to bound the probability over the entire tree.
Assume that there exist bad intervals of total length at least $\varepsilon N$,
then there must exist \emph{disjoint}
bad intervals of total length at least $\varepsilon N/2$,
as stated by the following lemma~\cite{schulman96}.

\begin{lemma}[\cite{schulman96}]\label{lem:intervals}
Let $\ell_1,\ell_2,\ldots,\ell_n$ be intervals on $\mathbb{N}$, of total length $X$.
Then there exists a set of indices $I\subseteq \{1, 2, \ldots, n\}$ such that
the intervals indexed by $I$ are disjoint, and their
total length is at least $X/2$. That is, for any $i,j\in I$, $\ell_i \cap \ell_j=\emptyset$, and
$
\sum_{i\in I}\lvert\ell_i\rvert \ge X/2
$.
\end{lemma}
\noindent The proof is given in~\cite{schulman96}.

There are at most $\sum_{j=\epsilon N/2}^N{N \choose j} \le 2^{N}$ ways to distribute these disjoint
intervals along the path from the \RTBC's root to $z$.
Using Lemma~\ref{lem:prob4smallHamming} and a union bound
we are assured that the probability of having
(any) node $u$ at the same depth as $z$ which imposes a bad interval of length $\ell$ is
 less than $(2d/|S|^{1-\alpha})^{\ell}$.
The probability for a specific pattern of disjoint bad intervals
to jointly occur is the multiplication of the probability for each interval to occur
(the intervals are independent since they are disjoint).
According to the above, for large enough $S$,
the probability for a \RTBC to be $(\varepsilon,\alpha)$-bad is bounded by
\begin{align*}
\Pr[\text{ \RTBC is $(\varepsilon,\alpha)$-bad }] &\le \sum_{z}
\sum_{\stackrel{\ell_1,\ell_2, \ldots \text{ disjoint,}}{\text{ of length} \ge \eps N/2}}
\prod_i (2d/|S|^{1-\alpha})^{\ell_i}  \\ &
\le d^N \cdot 2^N \cdot (2d/|S|^{1-\alpha})^{\sum_i \ell_i} \le (2d)^N (2d/|S|^{1-\alpha})^{\eps N/2}
\end{align*}
which is exponentially small in $N$ for $|S|> (2d\cdot(2d)^{2/\eps})^{1/(1-\alpha)}$.
\end{proof}

\subsection{Construction of a Pseudo-\RTBC Using Cryptographic Assumptions}

Using conventional cryptographic assumptions
and settings one can easily build a pseudo-\RTBC which can not be
distinguished from a truly random \RTBC.
In order to construct a pseudo $d$-ary \RTBC of depth $n$,
we assume the existence of a
family of pseudo-random functions (PRF)~\cite{GGM86}, $f_\lambda: \{0,1\}^* \to S$,
which can be computed efficiently.
The user randomly chooses a seed $\lambda$
of length $\kappa$,
and labels the arc $i$ with the label $f_\lambda(i)$.
When a pseudo-\RTBC is used to communicate between several
users, they all share the same seed $\lambda$.

\begin{lemma}
Let RT be a truly random \RTBC and let PRT be a pseudo-\RTBC, then for
any  algorithm $\mathcal A$ which is polynomial in $\kappa$,
$$
\left |\Pr \left [ \mathcal A^{\text{RT}}=1 \right] - \Pr [\mathcal A^{\text{PRT}}=1] \right | < neg(\kappa )\text{.} 
$$
\end{lemma}
\begin{proof}
Otherwise, $\mathcal A$ is a method to distinguish the
pseudo-random function $f_\lambda$ used to label the
\RTBC from a truly random function, in contradiction to it
being a pseudo-random function.
\end{proof}
A memoryless BSC channel can be considered as a (very restricted) polynomial-time algorithm.
The seed $\lambda$ for the PRF can be chosen by one party,
encoded using any good error-correcting code, and sent to the
other party at the start of the protocol.  Note that since all parties
are honest, and the channel is efficiently simulatable, there is no
need to hide the PRF's seed.
It follows that we can replace any use of \RTBC
with a pseudo-\RTBC, affecting the probabilities with only a negligible factor.

\section{Details of Theorem~\ref{thm:BR}}\label{app:BR}
We now prove Theorem~\ref{thm:BR}.
The proof follows the analysis of Braverman and Rao~\cite{BR10} in a straightforward way,
assuming the tree code in use is $(\epsilon_1,1-\epsilon_2)$-potent (that is, $\alpha=1-\epsilon_2$).

In~\cite{BR10} the users consider $\pi$ as a binary tree $\cal T$. Each path in the tree describes a possible transcript of $\pi$, where  odd levels describe party A's outputs and even levels describe B's outputs. The users use a good tree code to communicate the vertices of $\cal T$ according to their inputs.

Assume that at time $t$ user $A$ sends $a_t$ and let $a'_t$ be the label received at B's side
(similarly, User B sends $b_t$, etc.).
Upon receiving $a'_t$, user B decodes the received string $a'_1, \ldots, a'_t$ and obtains a possible transcript of $\pi$, from which he can compute his next step in $\pi$. This process is repeated for $R=\lceil T/\epsilon_2\rceil$ times.

Let $D(a'_1,\ldots, a'_t)$ denote a set of vertices in $\cal T$ described by decoding the received string. We denote with $m(i)$ the largest number such that the first $m(i)$ symbols of $D(a'_1, \ldots, a'_i)$ are equal to $a_1, \ldots, a_i$ and the first $m(i)$ symbols of $D(b'_1, \ldots, b'_i)$ are equal to $b_1, \ldots, b_i$.

Define ${\cal N}(i,j)$ to be the number of transmission errors in the $[i,j]$ interval of the simulation (for both users). In the analysis of~\cite{BR10}, a lower bound on the number of error in case that the simulation fails. We now show that using a $(\epsilon_1,1-\epsilon_2)$-\potent tree, the lower bound
changes by at most $\epsilon_1$.

The analysis of~\cite{BR10} begins by considering a simpler simulation in which the alphabet
size might be polynomial, and then extends the result to a constant alphabet size in a straightforward way. In order to ease the proof, we show that the theorem holds for the simple protocol with polynomial alphabet. Extending the result to the constant-alphabet protocol is immediate.

\begin{proof} (\textbf{Theorem~\ref{thm:BR}}.)
We redefine the quantity $\cal N$ to allow us consider possible errors caused by the tree in addition to   channel errors. Let ${\cal N}(i,j,d)$  be the number of communication errors between rounds $i$ and $j$, assuming that the total length of bad intervals along the paths $a_i,\ldots,a_j$ and $b_i,\ldots,b_j$ in the potent tree, is at most $d$.

\begin{lemma}[replacing lemma 4 of \cite{BR10}]\label{lem:4}
${\cal N}(m(i)+1,i,d)\ge (1-\epsilon_2)(i-m(i))/2-d$
\end{lemma}
\begin{proof}
Without loss of generality, we assume that the $m(i)+1$ symbol in $D(a'_1,\ldots,a'_i)$ differs from $a_{m(i)+1}$. Consider two cases. If the node $a_i$ is not $\alpha$-bad,
then the only way to get a decoding error of magnitude $l=(i-(m(i))$ is if at least
$\alpha l/2=(1-\epsilon_2)(i-m(i))/2$ communication errors have happened (this is identical to~\cite{BR10}).

In the second case, the node $a_i$ is $\alpha$-bad.
If $i-m(i) \le  d$ the lemma is trivial.
Otherwise, $a_i$ must be an $\alpha$-bad node of maximal length at most $d$.
$\Delta(a_1\cdots a_i,a'_1\cdots a'_i) \ge \alpha (i-m(i))$ and again such a decoding error
implies at least $(1-\epsilon_2)(i-m(i))/2$ communication errors.
\end{proof}

The quantity $t(i)$ is defined by~\cite{BR10} as the smallest round $j$ such that both users announced 
the first $i$ edges of $\cal T$ within their transmisssions. The following Lemma is stated in~\cite{BR10}.
\begin{lemma}[Lemma~5 of~\cite{BR10}]\label{lem:5}
For $i\ge 0, k\ge 1$, if $i+1<t(k)$, then $m(i) < t(k-1)$
\end{lemma}
\noindent The proof of this lemma is independent of the tree code in use, and thus it is valid for simulation with potent tree as well.

Last, we show the following lower bound on the number of errors.
\begin{lemma}[replacing lemma 6 of \cite{BR10}]\label{lem:6}
For $i\ge-1, k\ge 0$, if $i+1<t(k)$, then ${\cal N}(1,i,d)\ge(i-k+1)(1-\epsilon_2)/2-d$
\end{lemma}

\begin{proof}
We prove by induction. ${\cal N}(1,i,d)={\cal N}(1,m(i),x)+{\cal N}(m(i)+1,i,d-x)$
assuming that the total length of the imposed bad-intervals between
rounds $1$ and $m(i)$ (that is, along the paths $a_1,\ldots,a_{m(i)}$ and $b_1,\ldots,b_{m(i)}$)
is exactly $x$, $0\le x \le d$.
Lemma~\ref{lem:4} guarantees that ${\cal N}{(m(i)+1,i,d-x)}\ge (1-\epsilon_2)(i-m(i))/2-(d-x)$.
By Lemma~\ref{lem:5}, $m(i)<t(k-1)$ and we can use the induction hypothesis on the first part, which gives
${\cal N}(1,m(i)-1,x)\ge ((m(i)-1) - (k-1) +1)(1-\epsilon_2)/2-x$. Summing these two bounds proves the lemma.
\end{proof}
\noindent Note that the in the case of a good tree code, $d=0$, which gives exactly Lemma~6 of~\cite{BR10}. With a \potent tree, $d\le2\epsilon_1N$ which reduces the maximal error rate by $2\epsilon_1$.

In a similar way Lemma~8 of~\cite{BR10} can be adapted to potent trees, which completes the proof of Theorem~\ref{thm:BR}, by setting $\epsilon \ge \epsilon_1/2+\epsilon_2$.

\end{proof}

\end{document}